\documentclass[conference]{IEEEtran}
\usepackage[ruled,vlined]{algorithm2e}
\usepackage{xspace}
\usepackage{wrapfig}
\usepackage{subcaption}
\usepackage{tikz}
\usepackage{todonotes}
\usepackage{listings}
\usepackage{multirow}
\usepackage{amsthm}
\usepackage{subcaption}
\usepackage{microtype}
\usepackage{amsfonts}
\usepackage{balance}
\usepackage{url}

\SetCommentSty{mycommfont}

\usetikzlibrary{positioning, automata, shapes.arrows, calc, shapes, arrows}
\usepackage{amsmath}

\newcommand{\hide}[1]{}

\newcommand{\smo}{SyMO\xspace}
\newcommand{\symo}{\smo}
\newcommand{\smot}{SMTO\xspace}
\newcommand{\smto}{\smot}

\newcommand{\toolname}{Delphi\xspace}
\newcommand{\orint}{{\mathcal{J}}}
\newcommand{\freeorint}{{\mathcal{I}}}
\newcommand{\satformula}{\rho}
\newcommand{\quantifiedformula}{\rho}
\newcommand{\syntformula}{\phi}
\newcommand{\equality}{\approx}

\newcommand{\ointerface}[3]{(#1,#2,#3,\emptyset)}
\newcommand{\freeointerface}[4]{(#1,#2,#3,#4)}
\newcommand{\CDOT}{{{\cdot}}}
\newcommand{\mat}[1]{{#1}}

\newtheorem{definition}{Definition}[section]
\newtheorem{theorem}{Theorem}[section]

\lstdefinelanguage{alg}{
  sensitive = true,
  keywords={algorithm, input, output, return, if, else, for, while},
  numberstyle=\footnotesize,
  numbersep=9pt,
  showstringspaces=false,
  breaklines=true,
  comment=[l]{//},
  morecomment=[s]{/*}{*/}
}

\lstdefinestyle{grammar}{
  belowcaptionskip=1\baselineskip,
  basicstyle=\rmfamily\mdseries\footnotesize,
  breaklines=true,
  language=alg,
  xleftmargin=\parindent,
  showstringspaces=false,
  mathescape=true,
  numberstyle=\tiny,
  keywordstyle=\bfseries
}

\definecolor{codegreen}{RGB}{51,135,97}
\definecolor{codepurple}{RGB}{93, 49, 122}
\definecolor{codered}{RGB}{128, 18, 45}
\definecolor{backcolour}{rgb}{1.0,1.0,1.0}

\lstdefinelanguage{smt}{
  sensitive = false,
  keywords=[1]{declare, fun, synth, rec, check, assert, sat, blocking, define, constraint},
  keywordstyle=[1]\color{codepurple},
  keywords=[2]{div, mod, ite, and},
  keywordstyle=[2]\color{codegreen},
  keywords=[3]{Int, Bool},
  keywordstyle=[3]\color{codered},
  numbers=none,
  stepnumber=1,
  numbersep=8pt,
  showstringspaces=false,
  breaklines=true,
  comment=[l]{;},
  basicstyle=\footnotesize
}

\newcounter{myctr}

\newenvironment{myitemize}{\begin{list}{$\bullet$}
{\setlength{\topsep}{1mm}\setlength{\itemsep}{0.25mm}
\setlength{\parsep}{0.1mm}
\setlength{\itemindent}{0mm}\setlength{\partopsep}{0mm}
\setlength{\labelwidth}{15mm}
\setlength{\leftmargin}{4mm}}}{\end{list}}

\newcommand{\evalsatformula}{\mu}
\newcommand{\nf}[1]{{#1}{\downarrow}}
\newcommand{\teq}{\approx}

\newcommand{\modelrestrict}[2]{ #1\mid_{#2}}
\usepackage{microtype}
\usepackage{float}
\usepackage{dsfont}
\usepackage{latexsym}
\title{Satisfiability and Synthesis Modulo Oracles}
\author{\IEEEauthorblockN{Elizabeth Polgreen}
\IEEEauthorblockA{University of Edinburgh and UC Berkeley}
\and
\IEEEauthorblockN{Andrew Reynolds}
\IEEEauthorblockA{University of Iowa}
\and
\IEEEauthorblockN{Sanjit A.~Seshia}
\IEEEauthorblockA{UC Berkeley}}

\begin{document}

\maketitle
\begin{abstract}
In classic program synthesis algorithms, such as counterexample-guided inductive synthesis (CEGIS), the algorithms alternate between a synthesis phase and an oracle (verification) phase. Many synthesis algorithms use a white-box oracle based on satisfiability modulo theory (SMT) solvers to provide counterexamples. But what if a white-box oracle is either
not available or not easy to work with? 
We present a framework for solving a general class of oracle-guided synthesis problems which we term {\em synthesis modulo oracles}. In this setting, oracles may be black boxes with a query-response interface defined by the synthesis problem. As a necessary component of this framework, we also formalize the problem of {\em satisfiability modulo theories and oracles}, and present an algorithm for solving this problem.
% We show that our algorithm for solving synthesis modulo oracles is expressive enough to execute several algorithms for synthesis and verification including CEGIS and ICE learning. 
We implement a prototype solver for satisfiability and synthesis modulo oracles and demonstrate that, by using oracles that execute functions not easily modeled in SMT-constraints, such as recursive functions or oracles that incorporate compilation and execution of code, \smto and \smo are able to solve problems beyond the abilities of standard SMT and synthesis solvers.
\end{abstract}
%==============================================
\section{Introduction}
\label{sec:intro}
% \freeointerface{a,b,c,d}
A common formulation of program synthesis is to find a program, from
a specified class of programs,
that meets some correctness specification~\cite{sygus}. 
Classically, this is encoded as the 2nd-order logic formula 
$\exists \vec{f} \forall \vec{x} \,\, \syntformula$, 
where $\vec{f}$ is a set of target functions to be synthesized, 
$\vec{x}$ is a set of 0-ary symbols, and $\syntformula$ is a quantifier-free formula in a logical
theory (or combination of theories) $T$. 
A tuple of functions $\vec{f^*}$
satisfies the semantic restrictions if the formula 
$\forall \vec{x} \,\, \syntformula$ is valid in $T$ when the tuple is substituted for $\vec{f}$ in
$\syntformula$. Many problems are specified in this form, and the SyGuS-IF format~\cite{sygus-if} is 
one way of specifying such syntax-guided synthesis (SyGuS) problems.

Whilst powerful, this format is restrictive in one key way: it requires the correctness condition 
to be specified as a satisfiability modulo theories (SMT)~\cite{barrett-smtbookch09} formula in a fully white-box manner.
Put another way, SyGuS problems must be specified with static constraints before the solving process begins. This limits the problems that can be specified, as well as the oracles that can be used to guide the
% provide only a single type of oracle, an {\em SMT-based correctness oracle},
% which determines whether a candidate function (tuple) satisfies the constraints and returns a counterexample if not. 
% This limits the problems that can be specified as well as the oracles that can be used to 
search. For example, if one wants to synthesize (parts of) a protocol whose correctness needs to be
checked by a temporal logic model checker (e.g.~\cite{udupa-pldi13}), 
such a model-checking oracle cannot be directly invoked
within a general-purpose SyGuS solver and instead requires creating a custom solver.

Similarly, SMT solvers, used widely in verification and synthesis, require their input to be encoded as a logical formula prior to the initiation of solving. Whilst the language of SMT-LIB is powerful and expressive, many formulas are challenging for current SMT solvers to reason about; e.g., as in Figure~\ref{fig:prime}, finding a prime factorization of a given number.
%determining if a number is a prime number. 
Here it would be desirable to abstract this reasoning to an external oracle, rather than rely on the SMT solver's ability to reason about recursive functions. 
\begin{figure}
\begin{lstlisting}[language=smt]
(define-fun-rec isPrimeRec ((a Int) (b Int)) Bool
  (ite (> b (div a 2)) true
    (ite (= (mod a b) 0) 
      false
      (isPrimeRec a (+ b 1)))))

(define-fun isPrime ((a Int)) Bool
  (ite (<= a 1)
    false
    (isPrimeRec a 2)))

(assert (and (isPrime f1)(isPrime f2)(isPrime f3)))
(assert (= (* f1 f2 f3) 76))
\end{lstlisting}
\caption{SMT problem fragment: find prime factors of $76$. Unsolved by CVC4 v1.9. Solved by \smto using isPrime oracle in $<1$s \label{fig:prime}}
\end{figure}
% Additionally, it restricts the solution approach to be variants of a particular type of
% Typically the solution approach for SyGuS problems is thus some variable of a particular type of
% {\em counterexample-guided inductive synthesis} (CEGIS), which is 
% illustrated in Figure~\ref{fig:cegis}, wherein a synthesis phase sends queries to a verification oracle 
% in the form of a candidate program, and the verification phase response with a single point counterexample if the program is incorrect. 
% % \begin{wrapfigure}{R}{0.4\textwidth}
% \begin{figure}
% \vspace*{-5mm}
% \begin{center}
% \includegraphics[width=0.3\textwidth]{figures2/cegis.png}
% \caption{CEGIS with an SMT-based verification oracle returning counterexamples\label{fig:cegis}}
% \end{center}
% \vspace*{-5mm}
% \end{figure}
% % \end{wrapfigure}

This motivates our introduction of oracles to synthesis and SMT solving. 
Oracles can be black-box implementations that can be queried based on a pre-defined interface of query and response types. 
Examples of oracles could be components of systems that are too large and complex to analyze (but which can be executed on inputs) or external verification engines solving verification queries beyond SMT solving. 

Prior work has set out a theoretical framework expressing synthesis algorithms as oracle-guided inductive synthesis~\cite{ogis-theory}, where a learner interacts with an oracle via a pre-defined oracle interface. %
However, this work does not give a general algorithmic approach to solve
oracle-guided synthesis problems or demonstrate the framework on practical applications.
%the oracle query and response types are highly specialized and must be solved with specific tailored algorithms, and the paper does not aim to present a practical application of this theory. 
An important contribution we make in this work is to give {\em a unified algorithmic approach to solving
oracle-guided synthesis} problems, termed \smo. 
The \smo approach is based on a key insight: that query and response types
can be associated with two types of logical formulas: {\em verification assumptions} and
{\em synthesis constraints}. The former provides a way to encode restrictions on black-box oracle
behavior into an SMT formula, whereas the latter provide a way for oracles to guide the search of
the synthesizer.

\begin{figure}
\centering
\begin{subfigure}{.21\textwidth}
  \centering
  \includegraphics[width=\textwidth]{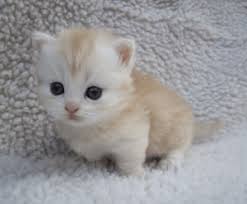}
  \caption{Original image}
\end{subfigure}%
\begin{subfigure}{.21\textwidth}
  \centering
  \includegraphics[width=\textwidth]{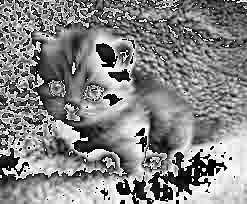}
  \caption{Target image}
\end{subfigure}
\caption{Image manipulation: transformations synthesized by \smo in $<1$ sec.}
\label{fig:invertcat}
\end{figure}

In order to explain the use-case for assumptions, let us first introduce \emph{oracle function symbols} and \emph{Satisfiability Modulo Theories and Oracles (\smot)}. Oracle function symbols are $n$-ary symbols whose behavior is associated with some oracle. Consider a quantifier-free formula $\satformula$ which contains an oracle function symbol $\theta$. \smot looks for a satisfying assignment to the formula based on initially assuming $\theta$ is a universally quantified uninterpreted function (i.e., we look for a satisfying assignment that would work for any possible implementation of the oracle): $\forall \theta \satformula$. As we make calls to the oracle, we begin to learn more about its behavior, and we encode this behavior as assumptions $\alpha$, such that the formula becomes $\forall \theta \alpha \Rightarrow \satformula$. This is the primary use case for assumptions generated by oracles, they are used to constrain the behavior of oracle function symbols. 
In \smo, determining the correctness of a candidate function is an \smot problem, and
assumptions generated by oracles are used in the \smot solving process.

% \begin{wrapfigure}{r}{0.5\textwidth}
% \vspace*{-4mm}
%     \begin{center}
%     % \input{figures/ogis}
%     \includegraphics[width=0.5\textwidth]{figures2/ogis.png}
%     \caption{oracle-guided Inductive Synthesis\label{fig:ogis}}
%     \label{fig:OGIS}
%     \end{center}
%     \vspace*{-4mm}
% \end{wrapfigure}
As an exemplar of an existing oracle-guided synthesis algorithm that goes beyond the SMT-solver-based counterexample oracles, consider ICE-learning~\cite{ice} for invariant synthesis. ICE-learning uses three oracles: an oracle to provide positive examples (examples which \emph{should} be contained within the invariant); an oracle to provide negative examples (examples which \emph{should not} be contained within the invariant); and an oracle to provide implication examples (an example pair where if the first  element is contained within the invariant, both must be contained). Whilst it is possible to build some of these oracles using an SMT solver, it is often more effective to construct these oracles in other ways, for instance, the positive example oracle can simply execute the loop or system for which an invariant is being discovered and return the output.

We implement \smo in a prototype solver \toolname, and hint at its broad utility by demonstrating several applications including programming by example, synthesis of controllers for LTI systems, synthesizing image transformations (e.g., Figure~\ref{fig:invertcat}, and satisfiability problems that reason about primes (e.g., Figure~\ref{fig:prime}). The latter use cases illustrate the power of being able to incorporate oracles into \smo that are too complex to be modeled or for SMT solvers to reason about. 

To summarize, the main contributions of this paper are:
\begin{myitemize}
\item A formalization of the problem of satisfiability and synthesis modulo oracles (Sec.~\ref{sec:oracles});
\item A unifying algorithmic approach for solving these problems (Sec.~\ref{sec:firstorder} and Sec.~\ref{sec:synthesis});
\item Demonstration of how this approach can capture popular synthesis strategies from the literature (Sec.~\ref{sec:existing}), and
\item A prototype solver \toolname, and an experimental demonstration of the broad applicability of this framework (Sec.~\ref{sec:evaluation}).
\end{myitemize}
%     \item First, we give a formal specification of the problem of satisfiability and synthesis modulo oracles;
%     \item We present a unifying algorithmic framework for solving such problems;
%     \item We show that this framework is expressive enough to capture popular synthesis algorithms in the literature;
%     \item And finally, we provide a prototype solver for both satisfiability and synthesis modulo oracles, and a set of example benchmarks showing a demonstration of the applicability of this framework.
% \end{itemize}
%\vspace{1em}
%The paper is structured as follows: in Section~\ref{sec:background} we give necessary background; in Section~\ref{sec:oracles} we formally define oracles and their corresponding oracle interfaces; in Section~\ref{sec:firstorder} we describe the impact of oracles on the semantics of satisfiability problems, and give an algorithm for solving satisfiability modulo oracles; in Section~\ref{sec:synthesis} we define the impact of oracles on the correctness criteria for synthesis problems, and give the algorithmic framework for solving these problems. We show reduction to existing synthesis algorithms in Section~\ref{sec:existing}. Finally in section~\ref{sec:evaluation} we present the first solver for satisfiability and synthesis modulo oracles evaluated on a range of case studies. 

\subsubsection*{Related work} 
Almost all synthesis algorithms can be framed as some form of oracle-guided synthesis. Counterexample-guided inductive synthesis (CEGIS) is the original synthesis strategy used for Syntax-Guided Synthesis~\cite{sketching}, and uses a correctness oracle that returns counterexamples.
Further developments in synthesis typically fall into one of two categories.
The first comprises innovative search algorithms to search the space more efficiently; 
for instance, genetic algorithms~\cite{david-proganalysis}, reinforcement learning~\cite{naik-iclr19}, or partitioning the search space in creative ways~\cite{eusolver}.
The second category comprises extensions to the communication paradigm permitted between the synthesis and the verification phase. For instance, CEGIS modulo theories~\cite{cegist}, CEGIS(T), extends the oracle interface over standard CEGIS to permit responses in the form of a restricted set of constraints over constants in the candidate program, while  
CVC4's single-invocation algorithm~\cite{cvc4-si} makes full use of the white-box nature of the SMT solver oracle. Other work leverages the ability to classify counterexamples as positive or negative examples~\cite{pldi-representationinvs}. 
There are also notable algorithms in invariant synthesis based on innovative use of different query types~\cite{pldi-representationinvs,ice}. Our work has one key stand-out difference over these: in all of these algorithms, the correctness criteria must be specified as a logical formula, whereas in our framework we enable specification of the correctness criteria as a combination of a logical formula and calls to external oracles which may be opaque to the solver. 
Synthesis with distinguishing inputs~\cite{ogis} is an exception to this pattern and uses a specific set of three interacting black-box oracles, to solve the very specific problem of synthesis of loop-free programs from components. 
Our work differs from this and the previously-mentioned algorithms in that they are customized to use certain specific types of oracle queries, whereas, we give a ``meta-solver'' allowing any type of oracle query that can be formulated as either generating a constraint or an assumption in the form of a logical formula. 

The idea of satisfiability with black-boxes has been tackled before on work on abstracting functional components as uninterpreted/partially-interpreted functions (see, e.g.,~\cite{andraus-dac04,brady-memocode10,conditional-abstractions}), which use counterexample-guided abstraction refinement~\cite{cegar}. Here, components of a system are abstracted and then refined based on whether the abstraction is sufficiently detailed to prove a property. However, to do this, the full system must be provided as a white-box. The key contribution our work makes in this area is a framework allowing the use of black-box components that obey certain query-response interface constraints, where the refinement is dictated by these constraints and the black-box oracle interaction.

\section{Oracles}
\label{sec:oracles}

In this section, we introduce basic definitions and terminology for the rest of
the paper. We begin with some preliminaries about SMT and synthesis.

%============================================
\subsection{Preliminaries and Notation}
\label{sec:background}

We use the following basic notations throughout the paper.
If $e$ is an expression and $x$ is free in $e$, let $e\CDOT\{x \rightarrow t \}$ be the formula obtained from the formula $e$ by proper substitution of the variable $x$ by the variable $t$.

\subsubsection{Satisfiability Modulo Theories (SMT)}

The input to an SMT problem is a first-order logical formula $\quantifiedformula$.
We use $\equality$ to denote the (infix) equality predicate.
The task is to determine whether $\quantifiedformula$ is $T$-satisfiable or $T$-unsatisfiable, that is, satisfied by a model which restricts the interpretation of symbols in $\quantifiedformula$ based on a background theory $T$.
If $\quantifiedformula$ is satisfiable, a solver will usually return a model of $T$ that makes $\quantifiedformula$ true, which will include assignments to all free variables in $\quantifiedformula$.
We additionally say that a formula is $T$-valid if it is satisfied by \emph{all} models of $T$.

\subsubsection{Syntax-Guided Synthesis}
 
In syntax-guided synthesis, we are given a set of functions $\vec{f}$ to be synthesized,
associated languages of expressions $\vec{L} = L_1,\ldots,L_m$ (typically generated
by grammars),
and we seek to solve a formula of the form
$$
\exists \vec{f} \in \vec{L} \forall \vec{x}\,\, \syntformula$$ 
where $\vec{x} \equality x_1 \ldots x_n$ is a set of 0-ary symbols and $\syntformula$ is a quantifier-free formula in a background theory $T$. 
In some cases, the languages $L_i$ include all well-formed expressions in $T$ of
the same sort as $f_i$, and thus $L_i$ can be dropped from the problem.
A tuple of candidate functions $\vec{f^*}$
% $( \lambda X_1.\, t_1, \ldots, \lambda X_n.\, t_n)$
satisfies the semantic restrictions
for functions-to-synthesize $\vec{f}$
in conjecture $\exists \vec{f}\, \forall \vec{x} \, \syntformula$ 
in background theory $T$
if $\forall \vec{x} \, \syntformula$ is valid in $T$
when $\vec{f}$ are defined to be terms
whose semantics are given by the functions
$(\vec{f^*} )$~\cite{sygus,sygus-if}.

%============================================
\subsection{Basic Definitions}
\label{subsec:oracles}

We use the term \emph{oracle} to refer to a (possibly black-box) component that can be queried in a pre-defined way by the solver. An oracle interface defines how an oracle can be queried. 
This concept is borrowed from~\cite{ogis-theory}. We extend the definition of oracle interfaces to also provide the solver with information on the \emph{meaning} of the response, in the form of expressions that generate assumptions or constraints.

\begin{definition}[Oracle Interface]
An oracle interface $\freeorint$ is a tuple $\freeointerface{\vec{y}}{\vec{z}}{\alpha_{gen}}{\beta_{gen}}$ where:
\begin{itemize}
    \item $\vec{y}$ is a list of sorted variables, which we call the \emph{query domain} of the oracle interface;
    \item $\vec{z}$ is a list of sorted variables, which we call its \emph{response co-domain};
    \item $\alpha_{gen}$ is a formula whose free variables are a subset of $\vec{y}, \vec{z}$, which we call its \emph{assumption generator}; and
    \item $\beta_{gen}$ is a formula whose free variables are a subset of $\vec{y}, \vec{z}$, which we call its \emph{constraint generator}.
\end{itemize}
$\Box$
\end{definition}

We assume that all oracle interfaces have an associated oracle that implements their prescribed interface for \emph{values} of the input sort, and generates concrete values as output.
In particular, an oracle for an oracle interface of the above form accepts
a tuple of values with sorts matching $\vec{y}$, and returns a tuple of values with sorts matching $\vec{z}$.
It is important to note that the notion of a value is specific to a sort,
which we intentionally do not specify here.
In practice, we assume e.g. the standard values for the integer sort; we assume all closed lambda terms are values for higher-order sorts, and so on.

An oracle interface defines how assumptions and constraints can be given to a solver via calls to black-box oracles, as given by the following definition.

% don't touch the spacing around this figure! 
% \begin{wrapfigure}{r}{0.5\textwidth}
\begin{figure}
\begin{minipage}{0.45\textwidth}
% \vspace*{-15mm}
    \centering
   \[
            I = \begin{cases}
            Q&: (y_1 :\,\sigma_1), \ldots,( y_j :\,\sigma_j)\\
            R&: (z_1 :\,\sigma'_{1}), \ldots,( z_k :\,\sigma'_{k})\\
            \alpha_{gen}&: \textit{assumption generator}\\
            \beta_{gen} &: \textit{constraint generator}\\
            \end{cases}
 \]\end{minipage}
    \caption{Oracle interface}
    \label{fig:oracleinterface}
\end{figure}
% \end{wrapfigure}

\begin{definition}[Assumptions and Constraints Generated by an Oracle Interface]
Assume $\freeorint$ is an oracle interface 
of form $\freeointerface{\vec{y}}{\vec{z}}{\alpha_{gen}}{\beta_{gen}}$.
We say formula $\alpha_{gen}\CDOT\{\vec{y} \rightarrow \vec{c}, \vec{z} \rightarrow \vec{d}\}$ is an assumption generated by $\freeorint$
if calling its associated oracle
for input $\vec{c}$ results in output $\vec{d}$.
%where $\vec{c}$ and $\vec{d}$ are tuples of values of the same sort as $\vec{y}$ and $\vec{z}$ respectively.
In this case, we also say that
$\beta_{gen}\CDOT\{\vec{y} \rightarrow \vec{c}, \vec{z} \rightarrow \vec{d}\}$
is a constraint generated by $\freeorint$.
$\Box$
\end{definition}

We are now ready to define the main problems introduced by this paper.
In the following definition, we distinguish two kinds of function symbols:
{\em oracle function symbols}, which are given special semantics in the following definition;
all others we call {\em ordinary function symbols}.
As we describe in more detail in Section~\ref{sec:firstorder},
oracle function symbols allow us to incorporate function symbols that correspond directly to oracles in specifications and assertions.

\begin{definition}[Satisfiability Modulo Theories and Oracles]
\label{def:smto}
A satisfiability modulo oracles (SMTO) problem is a tuple $(\vec{f}, \vec{\theta}, \quantifiedformula, \vec{\freeorint})$, 
where $\vec{f}$ is a set of ordinary function symbols,
$\vec{\theta}$ is a set of oracle function symbols,
$\quantifiedformula$ is a formula in a background theory $T$ whose free function symbols are $\vec{f} \uplus \vec{\theta}$, and $\vec{\freeorint}$ is a set of oracle interfaces.
We say this input is:
\begin{itemize}
\item 
% unsatisfiable if $\exists \vec{\theta}. A \Rightarrow (\quantifiedformula \wedge B)$ is
unsatisfiable if $\exists \vec{f}. \exists \vec{\theta}. A \wedge \quantifiedformula \wedge B$ is
$T$-unsatisfiable,
\item
satisfiable if $\exists \vec{f}. \forall \vec{\theta}. A \Rightarrow (\quantifiedformula \wedge B)$ is $T$-satisfiable,
\end{itemize}
where, in each case, $A$ (resp. $B$) is a conjunction of assumptions (resp. constraints) generated by $\vec{\freeorint}$.
%where $A$ (resp. $B$) is a conjunction of assumptions (resp. constraints) generated by $\vec{\freeorint}$.
%a formula $\gamma$ derived from $(\quantifiedformula, \vec{\freeorint})$ that is unsatisfiable modulo $T$ and the result is \emph{fixed}, i.e., all formula derived from $(\gamma, \vec{\freeorint})$ are also unsatisfiable modulo $T$.
%The input is satisfiable if there exists a formula $\gamma$ derived from $(\quantifiedformula, \vec{\freeorint})$ such that $\neg \gamma$ is unsatisfiable modulo $T$, and all subsequent formula derived from $(\gamma, \vec{\freeorint})$ are also satisfiable.
$\Box$
\label{def:smto-sat}
\end{definition}

According to the above semantics, constraints are simply formulas that we conjoin together with the input formula.
Assumptions play a different role. In particular, they restrict the possible interpretations of $\vec{\theta}$ that are relevant. As they appear in the antecedent in our satisfiability criteria, values of $\vec{\theta}$ that do not satisfy our assumptions need not be considered when determining whether an \smot input is satisfiable.
As a consequence of the quantification of $\vec{\theta}$, by convention we will say a model $M$ for an \smot problem contains interpretations for function symbols in $\vec{f}$ only; the values for $\vec{\theta}$ need not be given.

It is important to note the role of the quantification for oracle symbols $\vec{\theta}$ in the above definition.
An \smot problem is unsatisfiable if the conjunction of assumptions, input formula, and constraints are unsatisfiable when treating $\vec{\theta}$ existentially, i.e. as uninterpreted functions.
%In other words, a formula is unsatisfiable if it is unsatisfiable while treating $\vec{\theta}$ as uninterpreted.
Conversely, an \smot problem is satisfiable only if there exists a model satisfying $(\quantifiedformula \wedge B)$ for \emph{all} interpretations of $\vec{\theta}$ for which our assumptions $A$ hold.

In the absence of restrictions on oracle interfaces $\vec{\freeorint}$,
an \smot problem can be both satisfiable and unsatisfiable, depending on the constraints and assumptions generated.
For instance, when $A$ becomes equivalent to false, the input is trivially both unsatisfiable and satisfiable.
However, in practice, we define a restricted fragment of \smot, for which this is not the case, and
we present a dedicated procedure for this fragment in Section~\ref{sec:firstorder}.
To define this fragment, we introduce the following definition.

\begin{definition}[Oracle Interface Defines Oracle Function Symbol]
An oracle interface $\orint$ \emph{defines} an oracle function symbol $\theta$
if it is of the form
$\ointerface{(y_1, \ldots y_j)}{(z)}{\theta(y_1, \ldots y_j) \equality z}$,
and its associated oracle $\mathcal{O}$ is functional.
In other words, calling the oracle interface generates an equality assumption of the form $\theta(y_1, \ldots y_j) \equality z$ only. 
$\Box$
\end{definition}

%We also may say that $\orint$ is associated with oracle function symbol $\theta$ if it defines it.
From here on, as a convention, we use $\orint$ to refer to an oracle interface that specifically defines an oracle function symbol, and $\freeorint$ to refer to a free oracle interface, i.e., an oracle interface which may not define an oracle function symbol.

\begin{definition}[Definition Fragment of SMTO]
\label{def:defsmot}
An SMTO problem $(\vec{f}, \vec{\theta}, \quantifiedformula, \vec{\orint})$
is in Definitional SMTO if and only if
$\vec{\theta} = (\theta_1, \ldots, \theta_n)$,
$\vec{\orint} = (\orint_1, \ldots, \orint_n)$,
and $\orint_i$ is an oracle interface that defines $\theta_i$ for $i = 1, \ldots, n$.
$\Box$
\end{definition}
Note that each oracle function symbol is defined by one and only one oracle interface. 

We are also interested in the problem of synthesis in the presence of oracle function symbols,
which we give in the following definition.

\begin{definition}[Synthesis Modulo Oracles]
A synthesis modulo oracles (\smo) problem is a tuple
$(\vec{f}, \vec{\theta}, \forall \vec{x}. \,\,\syntformula, \vec{\freeorint})$,
where $\vec{f}$ is a tuple of functions (which we refer to as the functions to synthesize),
$\vec{\theta}$ is a tuple of oracle function symbols,
$\forall \vec{x}. \,\,\syntformula$ is a formula is some background theory $T$ where $\syntformula$ is quantifier-free,
and $\vec{\freeorint}$ is a set of oracle interfaces.
 %A tuple of functions $\vec{f^*}$ satisfies the synthesis conjecture if $\exists \vec{x} \forall \vec{\theta}\,\,\neg \syntformula [\vec{f^*}/\vec{f}]$ is unsatisfiable for any model for $\theta_1\ldots \theta_p$ that is consistent with the external oracles.
A tuple of functions $\vec{f^*}$ is a solution for synthesis conjecture if
$(\vec{x}, \vec{\theta}, \neg \syntformula \CDOT \{\vec{f} \rightarrow \vec{f^*}\}, \vec{\freeorint})$ is unsatisfiable modulo theories and oracles.
$\Box$
\label{def:symo}
\end{definition}
Although not mentioned in the above definition, the synthesis modulo oracles problem may be combined with paradigms for synthesis that give additional constraints for $\vec{f}$ that are not captured by the specification, such as syntactic constraints in syntax-guided synthesis.
In Section~\ref{sec:synthesis},
we present an algorithm for a restricted form of \smo problems where the verification of candidate solutions $\vec{f^*}$ reduces to Definitional \smot.

% We introduce a restricted fragment of \smo, whereby $\freeorint$ may only generate constraints. We will show in Section~\ref{sec:synthesis} that this restriction reduces the verification condition of synthesis to the definition fragment of \smot. 
% \begin{definition}[Definition Fragment of \smo]
% A \smo problem $(\vec{f}, \vec{\theta}, \forall \vec{x}. \,\,\syntformula, \vec{\orint}, \vec{\freeorint})$ is in the definition fragment of \smo \emph{iff}
% $\alpha_{gen}=\emptyset$ for all 
% $\freeorint \in \vec{\freeorint}$, and $(\vec{f}, \vec{\theta}, \forall \vec{x}, \vec{\orint})$ is in the definition fragment of \smto. 
% $\Box$
% \end{definition}

\section{Satisfiability Modulo Theories and Oracles}
\label{sec:firstorder}
\label{sec:smto}

In this section, we describe our approach to solving inputs in the definition fragment of \smot, according to Definition~\ref{def:defsmot}.
First, we note a subtlety with respect to satisfiability of \smot problems in the definition fragment vs. the general problem. Namely that a problem cannot be both satisfiable and unsatisfiable, and once a result is obtained for Definitional \smot, the result will not change regardless of subsequent calls to the oracles. This is not true for the general \smot problem. In particular, note the following scenarios:
\paragraph{Conflicting Results} Assume that
$\exists \vec{f}. \exists \vec{\theta}. A_i \wedge \quantifiedformula \wedge B_i$ is $T$-unsatisfiable, where $A_i$ (resp. $B_i$ be the conjunction of assumptions (resp. constraints) obtained after $i$ calls to the oracles.
In unrestricted \smto, it is possible that $A_i$ alone is $T$-unsatisfiable, thus $\forall \vec{\theta} A_i \Rightarrow (\quantifiedformula \wedge B_i)$ is $T$-satisfiable and the problem is both satisfiable and unsatisfiable.
However, in Definitional \smto, it is impossible for $A_i$ alone to be unsatisfiable,
since all oracle interfaces defining oracle function symbols, which generate assumptions only of the form $\theta(\vec{y}) \equality z$ and the associated oracles are functional.

\paragraph{Non-fixed Results} Assume that $\exists \vec{f}. \forall \vec{\theta}. A_i \Rightarrow (\quantifiedformula \wedge B_i)$ is $T$-satisfiable,
where $A_i$ (resp. $B_i$ be the conjunction of assumptions (resp. constraints) obtained after $i$ calls to the oracles.
Thus, by Definition~\ref{def:smto}, our input is satisfiable.
In unrestricted \smot, it is possible for an oracle to later generate an additional constraint $\beta$ such that $\forall \vec{\theta} A_i \Rightarrow (\quantifiedformula \wedge B_i \wedge \beta)$ is $T$-unsatisfiable, thus invalidating our previous result of ``satisfiable''. 
However, in Definitional \smot, this cannot occur, since oracles that generate non-trivial constraints are not permitted.
It is trivial that once any \smot is unsatisfiable, it remains unsatisfiable.
Thus the satisfiability results for Definitional \smot, once obtained, are fixed.

% Recall that a problem is unsatisfiable when $\forall \vec{\theta} A \Rightarrow ( \wedge \B)$ is $T$-satisfiable. Now, suppose that we have 

\subsection{Algorithm for Definitional \smto}

\begin{figure}
    \centering
    \includegraphics[width=0.3\textwidth]{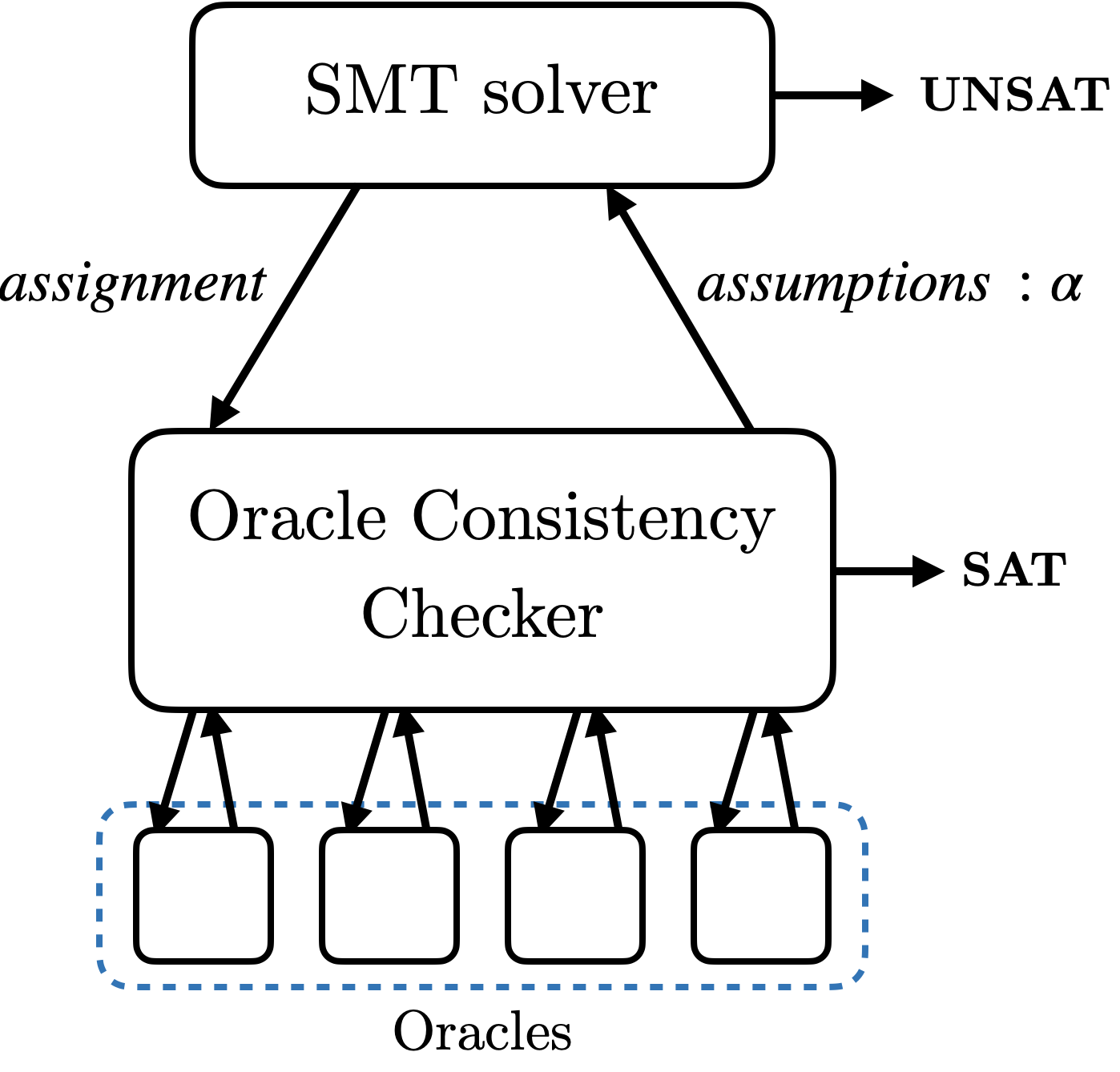}
    \caption{Satisfiability Modulo Oracle Solver}
    \label{fig:DPLLT_oracle}
\end{figure}

Our algorithm for Definitional \smot is illustrated in Figure~\ref{fig:DPLLT_oracle} and given as Algorithm~\ref{alg:sat}.
The algorithm maintains a dynamic set of assumptions $A$ generated by oracles.
In its main loop, we invoke an off-the-shelf SMT solver (which we denote $SMT$) on the conjunction of $\rho$ and our current assumptions $A$. If this returns UNSAT, then we return UNSAT along with the set of assumptions $A$ we have collected.
Otherwise, we obtain the model $M$ generated by the SMT solver from the previous call.

The rest of the algorithm (lines 8 to 20) invokes what we call the \emph{oracle consistency checker}.
Intuitively, this part of the algorithm checks whether our assumptions $A$ about $\vec{\theta}$ are consistent with the external implementation the oracle function symbols are associated with.

We use the following notation: we write $e[t]$ to denote an expression $e$ having a subterm $t$, and $e[s]$ to denote replacing that subterm with $s$.
We write $\nf{t}$ to denote the result of \emph{partially evaluating} term $t$.
For example, $\nf{(\theta(1+1)+1)}=\theta(2)+1$.

In the oracle consistency checker, we first construct the formula $\evalsatformula$ which replaces in $\rho$ all occurrences of ordinary function symbols $f$ with their value in the model $M$, and partially evaluate the result.
Thus, initially, $\evalsatformula$ is a formula whose free symbols are $\vec{\theta}$ only.
The inner loop (lines 9 to 17) incrementally simplifies this formula by calling external oracles to evaluate (concrete) applications of functions from $\vec{\theta}$.
In particular, while $\evalsatformula$ contains at least one application of a function from $\vec{\theta}$, that is, it is of the form $\evalsatformula[\theta_i(\vec{c})]$ where $\vec{c}$ is a vector of values. We know that such a term exists by induction, noting that an innermost application of a function from $\vec{\theta}$ must be applied to values.
We replace this term with the output $d$ obtained from the appropriate oracle.
The call to the oracle for input values $\vec{c}$ may already exist in $A$; otherwise, we call the oracle $\orint_i$ for this input and add this assumption to $A$.
After replacing the application with $d$, we partially evaluate the result and proceed.
In the end, if our formula $\evalsatformula$ is the formula $true$, the consistency check succeeds and we return SAT, along with the current set of assumptions and the model $M$.
We restrict the returned model so that it contains only interpretations for $\vec{f}$ and not $\vec{\theta}$, which we denote $\modelrestrict{M}{\vec{f}}$.
%
%That is, we first look for satisfying assignments to the formula $\satformula$ while over-approximating oracle function symbols as uninterpreted functions, and then checking whether this assignment is consistent with the oracles.
%Specifically, an SMT solver finds a model that makes the following formula true, assuming the oracle function symbols are uninterpreted functions:
%$\exists \vec{x},\exists \vec{\theta} \,\,\satformula.$
%
%
% \begin{figure}
%     \centering
%     \input{figures/dpll}
%     \caption{DPLL(T)}
%     \label{fig:DPLL}
% \end{figure}
%
%The corresponding oracles are called to check they are consistent with the model for $\vec{\theta}$.
%If the model was not consistent with the oracles, an oracle assumption is added to the 
%formula and we attempt to find a new model.
This process repeats until a model is found that is consistent with the oracles, or until the problem is shown to be unsatisfiable.

 We will now show that this intuitive approach is consistent with the previously defined semantics for \smto.

\begin{algorithm}[t]
\SetAlgoLined\LinesNumbered
\DontPrintSemicolon
\SetKwInOut{Input}{input}
\SetKwInOut{Output}{output}
\Input{$(\vec{f},\vec{\theta},\satformula,\vec{\orint})$}
\Output{UNSAT/SAT + assumptions $A$ + (model $M$)?}
  \SetKwFunction{SMTO}{SMTO}\SetKwFunction{consistent}{consistent}
  \SetKwProg{myalg}{Algorithm}{}{}
  \myalg{\smto{}}{
  $A \leftarrow true $\;
 \While{true}{
   \eIf{SMT($\satformula\wedge A$)=UNSAT}{
   \KwRet{UNSAT, A}
   }{
   Let $M$ be model for $\satformula\wedge A$ from $SMT$ \;
   Let $\evalsatformula$ be $\nf{(\satformula \CDOT \{ \vec{f} \rightarrow \vec{f}^M \})}$ \;
   \While{$\evalsatformula$ is of the form $\evalsatformula[ \theta_i( \vec c ) ]$}{
      \eIf{$(\theta_i( \vec c ) \teq d) \in A$ for some $d$}
      {
        $\evalsatformula \leftarrow \nf{\evalsatformula[d]}$
      }
      {
        Let $d = call\_oracle(\orint_i, \vec c)$ \;
        $A \leftarrow A \cup (\theta_i( \vec c ) \teq d)$ \;
        $\evalsatformula \leftarrow \nf{\evalsatformula[d]}$
      }
   }
   \If{$\evalsatformula$ is $true$}{
     \KwRet{SAT, A, $\modelrestrict{M}{\vec{f}}$}
   }
 }
}
}
 \caption{Satisfiability Modulo Oracles (\smot) \label{alg:sat}}
\end{algorithm}

\begin{theorem}[Correctness of \smot algorithm]
\label{thm:smto}
Algorithm~\ref{alg:sat} returns UNSAT (resp. SAT) iff the \smot problem 
$(\vec{f},\vec{\theta},\satformula,\vec{\orint})$ is unsatisfiable (resp. satisfiable) according to Definition~\ref{def:smto-sat}.
\end{theorem}
\begin{proof}
UNSAT case:
By definition, an \smto problem is unsatisfiable if $\exists \vec{f}. \exists \vec{\theta}. A \wedge \satformula$ is $T$-unsatisfiable, noting that for the definitional fragment of \smto, $B$ is empty.
Algorithm~\ref{alg:sat} returns UNSAT when the underlying SMT solver returns UNSAT on
the formula $\satformula \wedge A_0$ for some $A_0$.
Since $A_0$ is generated by oracles $\vec{\orint}$, it follows that our input is unsatisfiable.
%When $A$ is true, i.e., it does not contain a contradiction, and $\satformula \wedge A$ is UNSAT, it follows that $\exists \vec{\theta}A \wedge (\satformula)$ is also UNSAT. 

%$A$ is a conjunction of assumptions generated by the oracle interfaces $\vec{\orint}$. An assumption generated by $\orint_i$ is of the form $\theta_i(\vec{c})\approx d$. 
%Suppose $\orint_i$ generates two assumptions $\theta_i(\vec{c_1})\approx d_1$ and 
%$\theta_i(\vec{c_2})\approx d_2$.
%Since the corresponding oracle is functional, either $c_1 \approx c_2$ and $d_1 \approx d_2$ or $c_1 \neq c_2$, and so the conjunction of the assumptions cannot contain a contradiction. Assumptions generated by two different interfaces cannot contradiction because they are associated to different oracle function symbols. Thus $A$ cannot contain a contradiction, viz. the definition fragment of \smot is unsatisfiable when $A \wedge \satformula$ is unsatisfiable. $\Box$
SAT case:
By definition, an \smto problem is SAT \emph{iff} $\exists \vec{f}. \forall \vec{\theta}. A \Rightarrow \satformula$ is $T$-satisfiable for some $A$.
Algorithm~\ref{alg:sat} returns SAT when $\satformula \wedge A_0$ is SAT with model $M$ for some $A_0$, and when the oracle consistency check subsequently succeeds. Assume that the inner loop (lines 9 to 17) for this check ran $n$ times and that a superset $A_n$ of $A_0$ is returned as the set of assumptions on line 19.
We claim that $\modelrestrict{M}{\vec{f}}$ is a model for $\forall \vec{\theta}. A_n \Rightarrow \satformula$.
Let $M'$ be an arbitrary extension of $\modelrestrict{M}{\vec{f}}$ that satisfies $A_n$. Note that such an extension exists, since by definition of Definitional \smto, $A_n$ is a conjunction of equalities over distinct applications of $\vec{\theta}$.
%It suffices to show $M'$ satisfies $\satformula$.
Let $\evalsatformula_0, \evalsatformula_1, \ldots, \evalsatformula_n$ be the sequence of formulas such that $\evalsatformula_i$ corresponds to the value of $\evalsatformula$ after $i$ iterations of the loop on lines 9 to 17.
We show by induction on $i$, that $M'$ satisfies each $\evalsatformula_i$.
When $i=n$, $\evalsatformula_i$ is $true$ and the statement holds trivially.
For each $0 \leq i<n$, we have that $\evalsatformula_i$ is the result of
replacing an occurrence of $\theta(\vec{c})$ with $d$ in $\evalsatformula_{i-1}$ and partially evaluating the result, where $\theta(\vec{c}) \equality d \in A_n$.
Since $M'$ satisfies $\theta(\vec{c}) \equality d \in A_n$ and by the induction hypothesis satisfies $\evalsatformula_i$, it
satisfies $\evalsatformula_{i-1}$ as well.
Thus, $M'$ satisfies $\evalsatformula_0$, which is
$\nf{(\satformula \CDOT \{ \vec{f} \rightarrow \vec{f}^M \})}$.
%Since $M'$ extends $\modelrestrict{M}{\vec{f}}$, it satisfies $\rho$.
Thus, since $M'$ is an arbitrary extension of $\modelrestrict{M}{\vec{f}}$ satisfying $A_n$, we have that $\modelrestrict{M}{\vec{f}}$ satisfies
$\forall \vec{\theta}. A_n \Rightarrow \satformula$ and thus the input is indeed satisfiable.
\end{proof}

\begin{theorem}[Completeness for Decidable $T$ and Finite Oracle Domains]
Let background theory $T$ be decidable, and let the domain of all oracle function symbols be finite. In this case, Algorithm~\ref{alg:sat} terminates.
\end{theorem}
\vspace*{-1mm}
%\paragraph{Proof sketch} Termination is guaranteed 
\noindent
{\it Proof sketch:} Termination is guaranteed 
since the algorithm never repeats the same assignment to oracle inputs, and therefore, all input-output pairs for each oracle will be exhausted
eventually. At this point, the oracle function symbols can be replaced by interpreted functions (lookup tables),
and the formula reduces to one in the (decidable) background theory $T$.

Termination is not guaranteed in all background theories since it may be possible to write formulas where the number of input valuations to the oracle function symbols that must be enumerated is infinite. For example, this is possible in the theory of Linear Integer Arithmetic with an oracle function symbol with integer arguments.
%
% The following is an example of a formula that uses an oracle function symbol $\theta_1$ and where termination is not guaranteed if $x$ is an integer:
% $$
% \theta_1(x) = 1
% $$
% It may require the solver to try an infinite number of inputs to the oracle which corresponds to oracle function symbol $\theta_1$ before one is found that returns $1$. 
%
%However, in theories where the number of inputs to all oracle function symbols is finite, termination is guaranteed since eventually the behavior of the oracle for all inputs will be known and the oracle function symbols can be replaced by interpreted functions in the background theory. At this point the problem is decidable.  

\section{Synthesis Modulo Oracles}
\label{sec:synthesis}
\begin{algorithm}[t]
\SetAlgoLined
\LinesNumbered
\SetKwInOut{Input}{input}
\SetKwInOut{Output}{output}
\Input{$(\vec{f}, \vec{\theta}, \forall \vec{x} \syntformula, \vec{\orint}\uplus \vec{\freeorint})$}
\Output{solution $\vec{f^*}$ or no solution}
  $A \leftarrow \mathit{true}$ \tcp*[r]{conjunction of assumptions} 
  $S \leftarrow \mathit{true}$ \tcp*[r]{synthesis formula} 
 \While{true}{
  $\vec{f^*} \leftarrow$Synthesize( $\exists \vec{f} \,. S$) \;
  \eIf{$\vec{f^*} = \emptyset$}{
  \KwRet{no solution}\;
  }{
  $V \leftarrow A \wedge \neg \syntformula$  \tcp*[r]{verification formula} 
  $(r, \alpha, M)$ $\leftarrow$ SMTO$(\vec{x}, \vec{\theta}, V\CDOT\{\vec{f} \rightarrow \vec{f}^*\}, \vec{\orint}$) \; 
  \eIf{$r$=UNSAT}
  {
  \KwRet{$\vec{f^*}$} 
  }
  {
    % $\alpha', \beta'$ $\leftarrow$ call\_additional\_oracles($\vec{\freeorint},\syntformula, M$) \;
    $\beta$ $\leftarrow$ call\_additional\_oracles($\vec{\freeorint},\syntformula, M$) \;
    $A \leftarrow A \cup \alpha $ \;
    $S \leftarrow S \cup \syntformula \CDOT \{\vec{x} \rightarrow \vec{x}^M\} \cup \beta$\; 
    }
  }
  }
\caption{Synthesis Modulo Oracles \label{alg:synth}}
\end{algorithm}

% A synthesis modulo oracles (\smo) problem is a tuple
% $(\vec{f}, \vec{\theta}, \forall \vec{x}. \,\,\syntformula, \vec{\freeorint})$,
% where $\vec{f}$ is a tuple of functions (which we refer to as the functions to synthesize),
% $\vec{\theta}$ is a tuple of oracle function symbols,
% $\forall \vec{x}. \,\,\syntformula$ is a formula is some background theory $T$ where $\syntformula$ is quantifier-free,
% and $\vec{\freeorint}$ is a set of oracle interfaces.
%  %A tuple of functions $\vec{f^*}$ satisfies the synthesis conjecture if $\exists \vec{x} \forall \vec{\theta}\,\,\neg \syntformula [\vec{f^*}/\vec{f}]$ is unsatisfiable for any model for $\theta_1\ldots \theta_p$ that is consistent with the external oracles.
% A tuple of functions $\vec{f^*}$ is a solution for synthesis conjecture if
% $(\vec{x}, \vec{\theta}, \neg \syntformula \CDOT \{\vec{f} \rightarrow \vec{f^*}\}, \vec{\freeorint})$ is unsatisfiable modulo theories and oracles.

A \smo problem consists of: a tuple of functions to synthesize $\vec{f}$; a tuple of oracle function symbols $\vec{\theta}$; a specification in the form $\forall \vec{x}. \,\,\syntformula$, where $\syntformula$ is a quantifier-free formula in some background theory $T$, and a set of oracle interfaces $\vec{\freeorint}\uplus \vec{\orint}$.
We present an algorithm for a fragment of \smo, where the verification condition reduces to a Definitional \smto problem. To that end, we require that $\vec{\orint}$ is a set of oracle interfaces that define $\vec{\theta}$, and $\vec{\freeorint}$ is a set of oracle interfaces that only generate constraints, i.e., $\alpha_{gen}$ is empty.  We will show that these restrictions permit us to use the algorithm for Definitional \smto to check the correctness of a tuple of candidate functions in Theorem~\ref{thm:symo-corr}.

\subsection{Algorithm for Synthesis with Oracles}
\label{sec:alg}
We now proceed to describe an algorithm for solving synthesis problems using oracles, illustrated in Figure~\ref{fig:synthesis}. 
Within each iteration of the main loop, 
the algorithm is broken down into two phases: a {\em synthesis phase} and an {\em oracle phase}.
The former takes as input a synthesis formula $S$ which is incrementally updated over the course of the algorithm and returns a (tuple of) candidate solutions $\vec{f^*}$.
The latter makes a call to an underlying \smto solver for the verification formula $V$, which is a conjunction of the current set of assumptions $A$ we have accumulated via calls to oracles, and the negated conjecture $\neg \syntformula$.
%
% \begin{align*}
%     S &= \exists f_1 \ldots f_m.\,\,. true \\
%     V &= \exists x_1 \ldots x_n \,.\theta_1 \ldots \theta_p\,\,. \neg \syntformula.
% \end{align*}
%
%
In detail:
\begin{itemize}
    \item \textbf{Synthesis Phase:} 
    The algorithm first determines if there exists a set of candidate functions  $\vec{f}^*$ that satisfy the current synthesis formula $S$. 
    If so, the candidate functions are passed to the oracle phase. 
    \item \textbf{Oracle Phase I:} The oracle phase 
     calls the \smot solver as described in section~\ref{sec:firstorder} on the following Definitional \smot problem: $(\vec{x}, \vec{\theta}, V\CDOT\{\vec{f} \rightarrow \vec{f^*}\},\vec{\orint}$.
     If the \smot solver returns UNSAT, then $\vec{f}^*$ is a solution to the synthesis problem. 
     Otherwise, the \smot solver returns SAT, along with a set of new assumptions $\alpha$ and a model $M$.
     The assumptions $\alpha$ are appended to the set of overall assumptions $A$.
     Furthermore, an additional the constraint $\syntformula \CDOT \{\vec{x} \rightarrow \vec{x}^M\}$ is added to the current synthesis formula $S$.
     This formula can be seen as a counterexample-guided refinement, i.e. future candidate solutions must satisfy the overall specification for the values of $x$ in the model $M$ returned by the \smot solver.
     %In the process of solving the formula, any oracle assumption $\alpha$ that is generated is added to the verification formula, which becomes $\exists \vec{x} \, \forall \vec{\theta} (\alpha \Rightarrow \neg \syntformula)$. If the \smot solver returns SAT, then a counterexample constraint $\beta$ is generated, which is returned to the synthesis phase. 
     \item \textbf{Oracle Phase II: } As an additional step in the oracle phase, the solver may call any further oracles $\vec{\freeorint}$ and the constraints $\beta$ are passed to the synthesis formula. Note that the oracles in $\vec{\freeorint}$ generate constraints only and not assumptions.
\end{itemize}

\begin{table*}[]
    \centering
    \begin{tabular}{l | c | c }
    Query Type & Oracle Interface & Example synthesis algorithms\\ \hline
    \multicolumn{3}{c}{Constraint generating oracles} \\ \hline
    Membership & $\freeorint_{mem} \freeointerface{y_1, y_2, y}{z_b}{\emptyset}{z_b \Leftrightarrow f(y_1, y_2)=y}$ & Angluin's $L^*$~\cite{DBLP:lstar} \\
    Input-Output & $\freeorint_{io}\freeointerface{y_1, y_2, y}{z_b}{\emptyset}{z_b \Leftrightarrow f(y_1, y_2)\neq y}$ & Classic PBE \\
    Negative witness & $\freeorint_{neg}\freeointerface{\emptyset}{z_1, z_2,z}{\emptyset}{f(z_1, z_2)\neq z}$ & ICE-learning~\cite{ice} \\
    Positive witness & $\freeorint_{pos}\freeointerface{\emptyset}{z_1,z_2,,z}{\emptyset}{f(z_1, z_2)=z}$ & ICE-learning~\cite{ice}\\
    Implication & $\freeorint_{imp}\freeointerface{f^*}{z_1, z_2, z_1', z_2'}{\emptyset}{f(z_1, z_2)\Rightarrow f(z_1', z_2'}$ & ICE-learning~\cite{ice}\\
    Counterexample  & $\freeorint_{cex}\freeointerface{f^*}{\vec{z}}{\emptyset}{\syntformula\{\vec{x}{\rightarrow \vec{z}\}}}$ & Synthesis with validators~\cite{pldi-representationinvs} \\ 
    Distinguishing-input & $\freeorint_{di}\freeointerface{f^*}{z_1, z_2, z}{\emptyset}{f(z_1,z_2)=z}$ & Synthesis with distinguishing inputs~\cite{ogis}\\

    \hline
    \multicolumn{3}{c}{Constraint and assumption generating oracles} \\ \hline
    Correctness  & $\orint_{corr}\freeointerface{f^*}{z_b}{\theta(f^*)=z_b}{\emptyset}$ & ICE-learning~\cite{ice} \\ 
    Correctness with cex & $\orint_{ccex}\freeointerface{f^*}{z_b,\vec{z}}{\theta(f^*)=z_b}{\syntformula\{\vec{x}{\rightarrow \vec{z}\}}}$ & classic CEGIS~\cite{sketching}\\ 
    \end{tabular}
    \caption{Common oracle interfaces, illustrated for synthesizing a single function which takes a two inputs $f(x_1, x_2)$. $y$ indicates query variables, except where they are the candidate function, in which case we use $f^*$, and $z$ indicates response variables, where $z_b$ is a Boolean. \label{tab:orints}}
\end{table*}

\begin{figure}
%\begin{figure}
    \centering
    \includegraphics[width=0.35\textwidth]{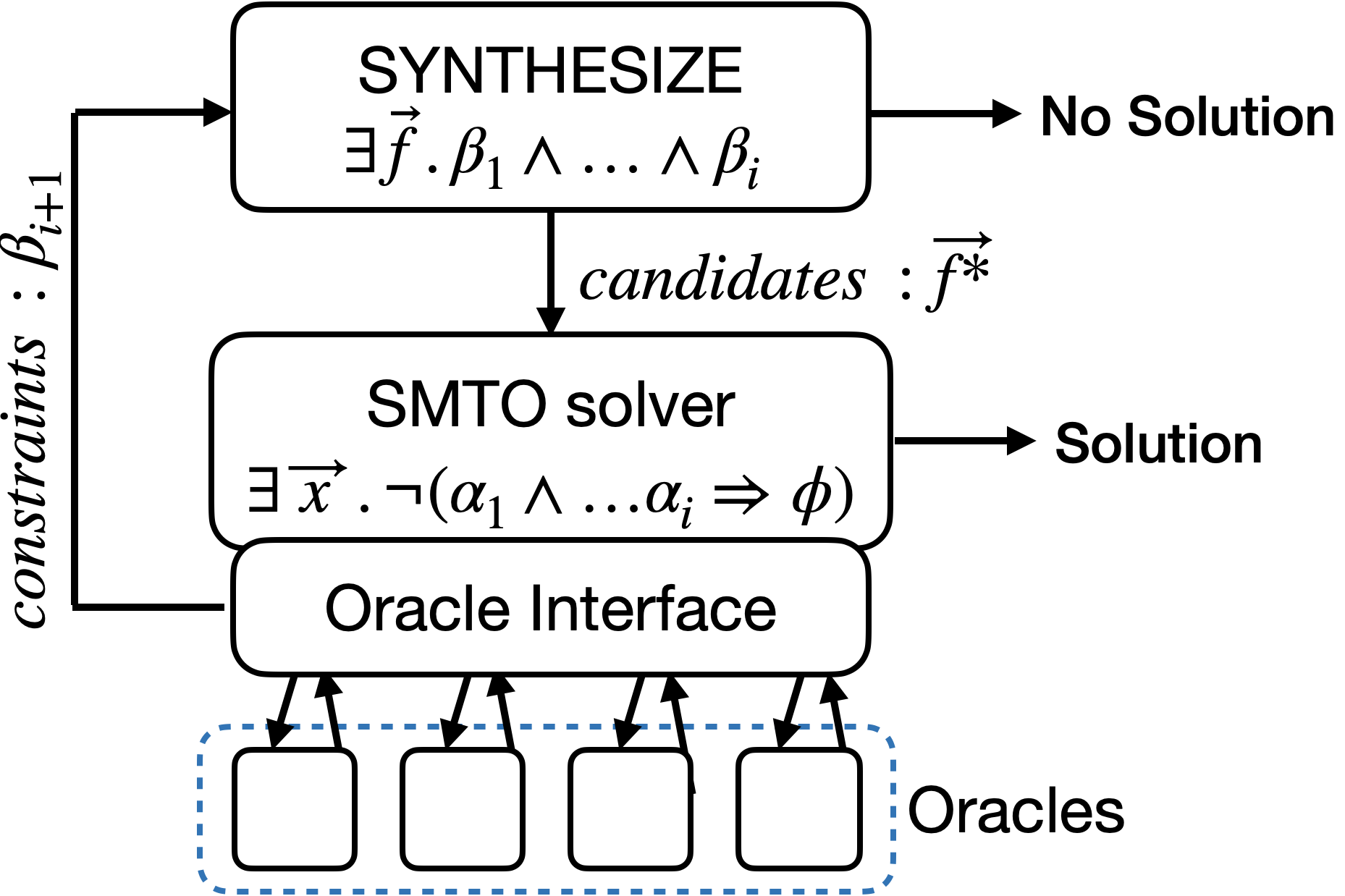}
    \caption{\smo Algorithm Illustration}
    \label{fig:synthesis}
%    \end{figure}
%\vspace*{-30mm}
 \end{figure}

\begin{theorem}[Soundness]
\label{thm:symo-corr}
If Algorithm~\ref{alg:synth} returns $\vec{f^*}$, then $\vec{f^*}$ is a valid solution for the \smo problem $(\vec{f}, \vec{\theta}, \forall \vec{x} \syntformula, \vec{\orint}\uplus \vec{\freeorint})$. %according to Definition~\ref{def:symo}.
\end{theorem}
\begin{proof}
According to Definition~\ref{def:symo},
a solution $\vec{f^*}$ is valid for our synthesis problem \emph{iff} $(\vec{x}, \vec{\theta}, \neg \syntformula \CDOT \{\vec{f} \rightarrow \vec{f^*}\}, \vec{\orint}\uplus \vec{\freeorint})$ is unsatisfiable modulo theories and oracles, i.e. when $\exists \vec{\theta} A \wedge (\neg \syntformula \CDOT \{\vec{f} \rightarrow \vec{f^*}\} \wedge B)$ is $T$-unsatisfiable for assumptions $A$ and constraints $B$ generated by oracle interfaces $\vec{\orint}\uplus \vec{\freeorint}$.
By definition, Algorithm~\ref{alg:synth} returns a solution if the underlying \smot solver finds that $(\vec{x}, \vec{\theta}, \neg \syntformula \CDOT \{\vec{f} \rightarrow \vec{f^*}\},\vec{\orint})$ is unsatisfiable modulo theories and oracles, i.e. $\exists \vec{\theta} A \wedge (\neg \syntformula \CDOT \{\vec{f} \rightarrow \vec{f^*}\})$ is $T$-unsatisfiable, which trivially implies that the above statement holds.
%If this formula is unsatisfiable, and since $\vec{\freeorint}$ generate only constraints, it follows that $\exists \vec{\theta} A \wedge (\neg \syntformula \CDOT \{\vec{f} \rightarrow \vec{f^*}\} \wedge B)$ is also unsatisfiable. 
Thus, and since the \smot solver is correct for UNSAT responses due to Theorem~\ref{thm:smto}, any solution returned by Alg.~\ref{alg:synth} is a valid solution.
\end{proof}

%\paragraph{Inferring inputs for additional oracles} 
\noindent
{\it Inferring inputs for additional oracles:} 
Although not described in detail in Algorithm~\ref{alg:synth}, we remark that an implementation may infer additional calls to oracles based on occurrences of terms in constraints from $\vec{\freeorint}$ and ground terms in $\syntformula$ under the current counterexample from $M$.
%The inputs for oracle function symbols are inferred by the satisfiability modulo oracles solver. For other oracles, the input values are inferred by mapping concrete values from the counterexample to applications of the synthesis function in $\syntformula$ and mapping oracle query inputs to synthesis function applications within constraint generators.
For example, if $f(7)$ appears in $\syntformula \CDOT \{\vec{x} \rightarrow \vec{x}^M\}$, and there exists an oracle interface with a single input $z$ and the generator $\beta_{gen}: f(z) \teq y$, we will call that oracle with the value $7$.
In general, inferring such inputs amounts to matching terms from constraint generators to concrete terms from $\syntformula \CDOT \{\vec{x} \rightarrow \vec{x}^M\}$.
Our implementation in Section~\ref{sec:evaluation} follows this principle.

\section{Instances of Synthesis Modulo Oracles}
%\section{Examples of Generality and Expressiveness}
\label{sec:existing}

 A number of different queries are categorized in work by Jha and Seshia~\cite{ogis-theory}. Briefly, these query types are
 \begin{itemize}
     \item membership queries: the oracle returns true \emph{iff} a given input-output pair is permitted by the specification
     \item input-output queries: the oracle returns the correct output for a given input
     \item positive/negative witness queries: the oracle returns a correct/incorrect input-output pair
     \item implication queries: given a candidate function which the specification demands is inductive, the oracle returns a counterexample-to-induction~\cite{ic3, ice}.
     \item Counterexample queries: given a candidate function, the oracle returns an input on which the function behaves incorrectly if it is able to find one
     \item Correctness queries: the oracle returns true \emph{iff} the candidate is correct
     \item Correctness with counterexample: the oracle returns true \emph{iff} the candidate is correct and a counterexample otherwise
     \item Distinguishing inputs: given a candidate function, the oracle checks if there exists another function that behaves the same on the set of inputs seen so far, but differently on a new input. If one exists, it returns the new input and its correct output. 
 \end{itemize}
% including queries such as membership queries, where the learner selects an example set of input-output pair(s) and asks if they are permitted by the specification $\syntformula$. 
We show the oracle interfaces for each of the classic query types in Table~\ref{tab:orints}. Full details are given in Appendix~\ref{sec:appendix-instances}. To give an intuition for how these interfaces are \emph{used} in a synthesizer: oracle function symbols, $\vec{\theta}$ are used to encapsulate correctness criteria (in full or in part) of $\vec{f}$, when correctness depends on external oracles. For instance, we could write a specification that states that $\theta(\vec{f})$ should return true, and assumptions generated by the oracle are used to determine when that specification is satisfied.
Constraints generated by oracles associated with $\vec{\freeorint}$ are used to guide the synthesizer. Synthesis algorithms typically alternate between a synthesis phase, which solves an approximation of the specification, and a verification phase which verifies the solution to the approximation against the full specification. For example, CEGIS~\cite{sketching} uses the counterexamples to construct this approximation in the form of constraints specifying the behavior of the function must be correct on the counterexamples obtained so far. 

Thus the synthesis modulo oracles framework we present is a flexible and general framework for program synthesis. \symo can implement any inductive synthesis algorithm, i.e., any synthesis algorithm where the synthesis phase of the algorithm iteratively increases the constraints over the synthesis function. Common synthesis algorithms are composed of combinations of the queries described in Table~\ref{tab:orints}. For instance, CEGIS~\cite{sketching} is \smo with a single counterexample-with-correctness interface $\orint_{ccex}$; ICE-learning~\cite{ice} uses $\orint_{corr}, \freeorint_{imp}, \freeorint_{pos}, \freeorint_{neg}$. For full details on these equivalences, see Appendix~\ref{sec:appendix-instances}.

\section{\toolname: a Satisfiability and Synthesis Modulo Oracles Solver}
\label{sec:evaluation}
We implement the algorithms described above in a prototype solver \toolname\footnote{\scriptsize link: \url{https:://github.com/polgreen/delphi}}.
The solver can use any SMT-lib compliant SMT solver as the sub-solver in the \smto algorithm, or bitblast to MiniSAT version 2.2~\cite{minisat}, and it can use any SyGuS-IF compliant synthesis solver in the synthesis phase of the \smo algorithm, or a symbolic synthesis encoding based on bitblasting. In the evaluation we report results using CVC4~\cite{cvc4} v1.9 in the synthesis phase and as the sub-solver for the \smto algorithm. 
The input format accepted by the solver is an extension of SMT-lib~\cite{smtlib} and SyGuS-IF~\cite{sygus-if}.

\subsection{Case Studies}
We aim to answer the following research questions: 
RQ1 -- when implementing a logical specification as an oracle executable, what is the overhead added compared to the oracle-free encoding?
RQ2 -- can \smto solve satisfiability problems beyond state-of-the-art SMT solvers? RQ3 -- can \symo solve synthesis problems beyond state-of-the-art SyGuS solvers? 
To that end, we evaluate \toolname on the following case studies.
\paragraph{Reasoning about primes} We convert a set of 12 educational mathematics problems~\cite{maths1} (Math) that reason about prime numbers, square numbers, and triangle numbers into SMT and \smto problems. We demonstrate that using an oracle to determine whether a number is a prime, a square or a triangle number is more efficient than the pure SMT encoding.
\paragraph{Image Processing}
    Given two images, we encode a synthesis problem to synthesize a pixel-by-pixel transformation between the two.
    Figure~\ref{fig:invertcat} shows an example transformation. 
    The \symo problem uses an oracle, shown in Figure~\ref{fig:imageoracle}, which loads two JPEG images of up to $256\times256$ pixels: the original image, and the target image. Given a candidate transformation function, it translates the function into C code, executes the compiled code on the original image and compares the result with the target image, and returns ``true'' if the two are identical. If the transformation is not correct, it selects a range of the incorrect pixels and returns constraints to the synthesizer that give the correct input-output behavior on those pixels. The goal of the synthesis engine is to generalize from few examples to the full image. The oracle-free encoding consists of an equality constraint per pixel. This is a simplification of the problem which assumes the image is given as a raw matrix and omits the JPEG file format decoder.
\paragraph{Digital Controller Synthesis} We encode the task of synthesizing floating-point controllers which guarantee bounded safety and stability of Linear Time Invariant systems~\cite{control-cegis}, using a range of time discretizations as both a \symo and an oracle-free synthesis problem. We use two oracles in the \symo problem: a correctness oracle to check the stability of the system using Eigen~\cite{eigen} (which avoids the need for the solver to reason about complex non-linear arithmetic). We note that these benchmarks do not necessarily have solutions for all time discretizations.
\paragraph{Programming by example} We encode PBE~\cite{sygus-comp} benchmarks as \symo problems using oracles that demonstrate the desired behavior of the function to be synthesized. These examples show that PBE benchmarks have a simple encoding in our framework.
% \hide{
% \paragraph{Invariant synthesis~\cite{sygus-comp}} We run \toolname on invariant synthesis benchmarks from the SyGuS competition and compare static constraints vs performance with positive, negative and implication oracles where possible (where we only provide positive witness oracles that unroll the transition relation where the transition relation is deterministic).}

\begin{figure}
    \centering
    \includegraphics[width=0.45\textwidth]{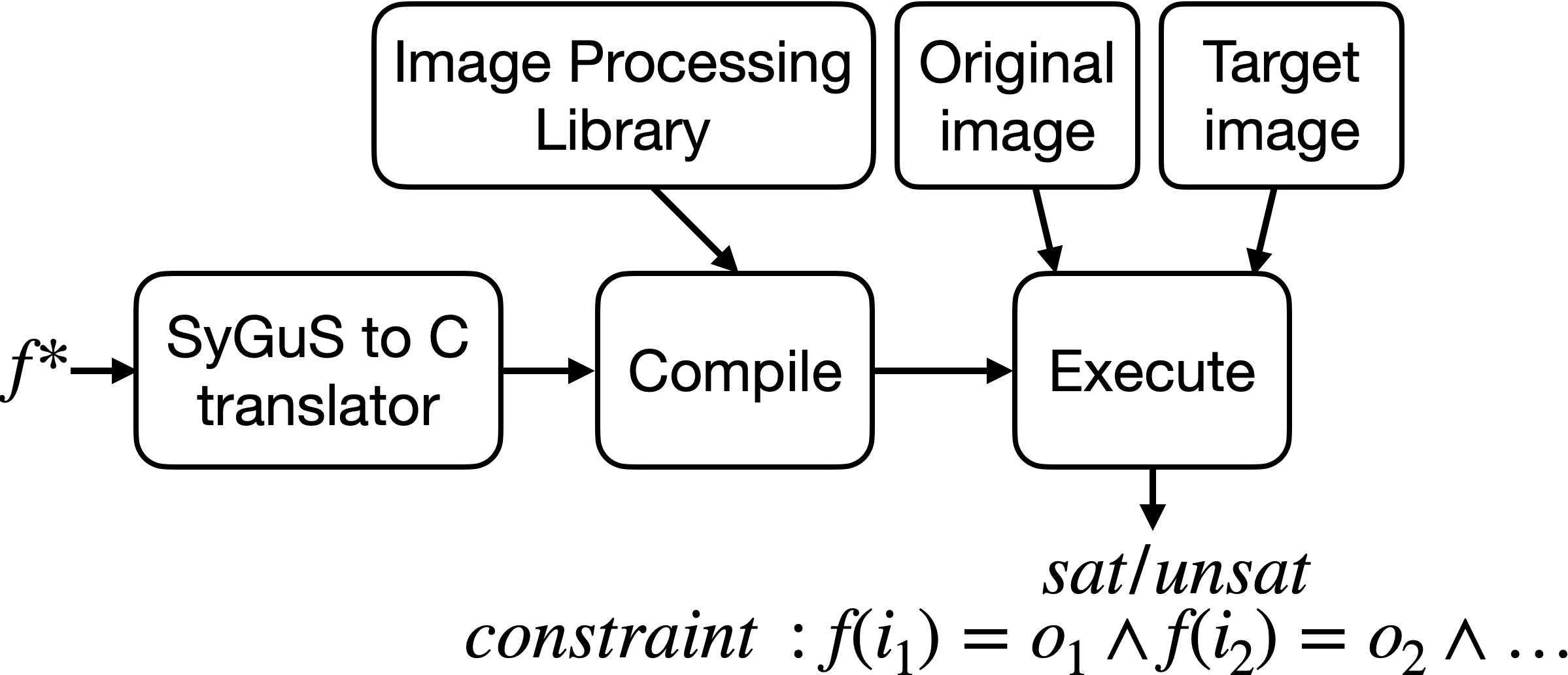}
    \caption{Oracle for image transformations}
    \label{fig:imageoracle}
\end{figure}

\begin{table}
    \begin{tabular}{l  |c | c | c  | c | c | c }
        & \multicolumn{2}{c|}{\toolname}  &\multicolumn{2}{c|}{\toolname} & \multicolumn{2}{|c}{CVC4} \\ 
        (encoding) & \multicolumn{2}{c|}{(oracles)} & \multicolumn{2}{c|}{(no oracles)} &\multicolumn{2}{|c}{(no oracles)}\\ \hline
         Benchmarks (\#)          & \# & t  & \# & t  & \# & t \\\hline
        Images(10) & 9 & 21.6s & 0 & -- & 0 & --  \\
        Math(12) & 9 & $<$0.5s &  1 & $<$0.2s &  5 & 2.2s   \\
        Control-stability(112) & 104 & 29.3s & -- & -- & 16 & 19.4s  \\
        Control-safety(112) & 31 & 59.9s &  0 & -- & 0 &  --  \\
        PBE(150) & 148 & 0.5s & 150 & 1.6s & 150 & $<$0.2s \\
        % ICE/CEGIS~\cite{sygus} () &  &  &  &   \\
    \end{tabular}
\caption{Comparison of \toolname and CVC4. \# is the number of benchmarks solved within the $600s$ timeout, and $t$ is the average run-time for solved benchmarks. The first column shows results on \symo and \smto problems, the second two columns show results on the equivalent oracle-free encodings. \label{tab:main_results}}
\end{table}

\subsection{Observations}
We report a summary of the results for these case-studies in Table~\ref{tab:main_results} and make the following observations:

\paragraph{RQ1} The overhead incurred by using oracles is small: performance on PBE problems encoded with oracles is similar to PBE problems encoded without oracles, with a small overhead incurred by calling external binaries. 
Given this low overhead, \smo would be amenable to integration with many more sophisticated synthesis search approaches~\cite{dillig-cdcl,cvc4sy,eusolver}.

\paragraph{RQ2} \toolname solves more educational mathematics questions than CVC4, demonstrating that \smto does enable SMT solvers to solve problems beyond the state-of-the-art by delegating challenging reasoning to an external oracle.

\paragraph{RQ3} \toolname solves control synthesis problems and image transformation problems that cannot be easily expressed as SyGuS and elude CVC4, demonstrating that \symo can solve synthesis problems beyond state-of-the-art solvers. When tackling the image transformation problems, \symo dynamically generates small numbers of informative constraints, rather than handling the full image at once.  

%
% \hide{Use of ICE oracles enables \toolname to solve an additional X invariant benchmarks, by providing more informative specifications.}

We also note that in many cases the encodings for \smo and \smto problems are more compact and (we believe) easier to write in comparison to pure SMT/SyGuS encodings. For instance, Figure~\ref{fig:prime} reduces to two assertions and a declaration of a single oracle function symbol.

\textbf{Future work:}
We see a lot of scope for future work on \smo. In particular, we plan to embed \smot solving into software verification tools; allowing the user to replace functions that are tricky to model with oracle function symbols. The key algorithmic developments we plan to explore in future work include developing more sophisticated synthesis strategies that decide when to call oracles based on the learned utility and cost of the oracles. An interesting part of future work will be to explore interfaces to oracles that provide \emph{syntactic} constraints, such as those used in~\cite{cegist,dillig-cdcl}, which will require the use of context-sensitive grammars in the synthesis phase.

\section{Conclusion}

We have presented a unifying framework for synthesis modulo oracles, identifying two key types of oracle query-response patterns: those that return constraints that can guide the synthesis phase and those that assert correctness. We proposed an algorithm for a meta-solver for solving synthesis modulo oracles, and, as a necessary part of this framework, we have formalized the problem of satisfiability modulo oracles. Our case studies demonstrate the flexibility of a reasoning engine that can incorporate oracles based on complex systems, which enables \smto and \smo to tackle problems beyond the abilities of state-of-the-art SMT and Synthesis solvers, and allows users to specify complex problems without building custom reasoning engines.

\hide{
\section*{Acknowledgments}

We thank Susmit Jha for his feedback on an earlier version of this paper.
This work was supported in part by NSF grants CNS-1739816 and CCF-1837132, by the DARPA LOGiCS project under contract FA8750-20-C-0156, by the iCyPhy center, and by gifts from Intel, Amazon, and Microsoft. 
}
\newpage
\balance 

%=============================================
\bibliographystyle{plain}
\bibliography{references}
 
%=============================================
\newpage
\appendix
\subsection{Instances of Synthesis Modulo Oracles}
%\section{Examples of Generality and Expressiveness}
\label{sec:appendix-instances}

The \symo framework we present is a flexible and general framework for program synthesis. \symo can implement any inductive synthesis algorithm, i.e., any synthesis algorithm where the synthesis phase of the algorithm iteratively increases the constraints over the synthesis function. 
% In the case of synthesizing a predicate, this means any synthesis algorithm where if the candidate function obtained at iteration 1, $f^*_1$, is true, then the candidate function obtained at iteration 2 is true, and, more generally, $f^*_{n} \implies f^*_{n+1}$.
%
%
Here we describe how, by providing specific oracles the algorithm describes will implement standard synthesis algorithms such as CEGIS~\cite{sketching}, ICE-learning~\cite{ice}, Synthesis with distinguishing inputs~\cite{ogis} and CEGIS(T)~\cite{cegist}. 
\subsubsection{CounterExample Guided Inductive Synthesis: }
Suppose we are solving a synthesis formula with a single variable $x$ and a single synthesis function $f$, where $f:\sigma \rightarrow \sigma'$. CEGIS consists of two phases, a synthesis phase that solves the formula $S=\exists f \forall x \in X_{cex}, \syntformula$, where $X_{cex}$ is a subset of all possible values of $x$, and a verification phase which solves the formula $V = \exists x \neg \syntformula$. There are two ways of implementing CEGIS in our framework. The first is simply to pass the full SMT-formula $\syntformula$ to the algorithm as is, without providing external oracles. 
%
%
%
% In the algorithm described above, we can see that the conjunction of constraints generated by the Satisfiability Modulo Oracles solver at each iteration is equivalent to the synthesis formula $\exists f \forall x \in I', \syntformula$, and the verification formula is consistent with the verification formula used in CEGIS. 
%
The second method is to replace the specification given to the oracle guided synthesis algorithm with $\exists f \forall \theta\,. \theta(f)$ and use an external correctness oracle with counterexamples, illustrated here for a task of synthesizing a function $f$, and receiving a candidate synthesis function $y:\sigma \rightarrow \sigma'$:
            \[
            I_{corr}  = 
            \begin{cases}
            Q&: (y \,\, (\sigma \rightarrow \sigma'))\\
            R&: (z_1 \,\sigma), (z_2 \,\, bool)\\
            \alpha_{gen}&: \theta(y)=z_2 \\
            \beta_{gen}&: \syntformula(z_1/x) \\
            \end{cases}
            \]
By inspecting the formula solved by the synthesis phase at each iteration, we can see that, after the first iteration, the synthesis formula are equisatisfiable if the sequence of counterexamples obtained is the same for both algorithms. 
\begin{table}[H] % formatting inserted temporarily to avoid Fig 7 getting split between pages
    \centering
    \begin{tabular}{l| c|c}
      iter. &  CEGIS & \smo with correctness oracle   \\\hline
        1   &  $X_{cex} = \emptyset$ & \\ 
            &  $\exists f. \exists x. \syntformula$ & $\exists f. \mathit{true}$ \\ \hline
        2   &  $X_{cex} = {c_1}$ & $\beta_1 = \syntformula(k_1/x)$  \\
            &. $\exists f. \forall x \in X_{cex} \,.\syntformula(x)$ &  $\exists f. \beta_1$ \\ \hline
        3   &  $X_{cex} = {c_1, c_2}$ & $\beta_2 = \syntformula(k_2/x)$  \\
            &. $\exists f. \forall x \in X_{cex} \,.\syntformula(x)$ &  $\exists f. \beta_1 \wedge \beta_2$ \\    \hline
        $\ldots$ & $\ldots$ &$\ldots$ \\    
    \end{tabular}
    \caption{Comparison of the synthesis formula at each iteration, showing that, if the same sequence of counterexamples is obtained, the synthesis formula are equisatisfiable at each step, i.e., \smo reduces to CEGIS.}
    \label{tab:my_label}
\end{table}

\subsubsection{ICE learning}
ICE learning~\cite{ice} is an algorithm for learning invariants based on using examples, counterexamples and implications. Recall the classic invariant synthesis problem is to find an invariant $inv$ such that:
\begin{align*}
\forall x, x' \in X. init(x)&\implies inv(x) \wedge \\
inv(x) \wedge trans(x,x') &\implies inv(x') \wedge \\
inv(x') &\implies \syntformula
\end{align*}
where $init$ defines some initial conditions, $trans$ defines a transition relation and $\syntformula$ is some property that should hold. ICE is an oracle guided synthesis algorithm, where, given a candidate $inv^*$, if the candidate is incorrect (i.e., violates the constraints listed above) the oracle can provide:
%\begin{itemize}
positive examples $E \subseteq X$, which are values for $x$ where $inv(x)$ should be $\mathit{true}$;
negative examples $C \subseteq X$, which are values for $x$ where $inv(x)$ should be $\mathit{false}$; and
 implications $I \subseteq X \times X$, which are values for $x$ and $x'$ such that $inv(x) \Rightarrow inv(x')$.
The learner then finds a candidate $inv$, using a symbolic encoding, such that 
\begin{align*}
 (\forall x \in E. inv(x)) \,\,\wedge \,\, \\
 (\forall x \in C. \neg inv(x))\,\,\wedge \,\, \\ 
 (\forall (x,x') \in I. inv(x) \Rightarrow inv(x')).
% (\forall (x,x') \in I. inv(x) \iff inv(x')).
\end{align*}

\begin{figure}
    \centering
    \includegraphics[width=0.5\textwidth]{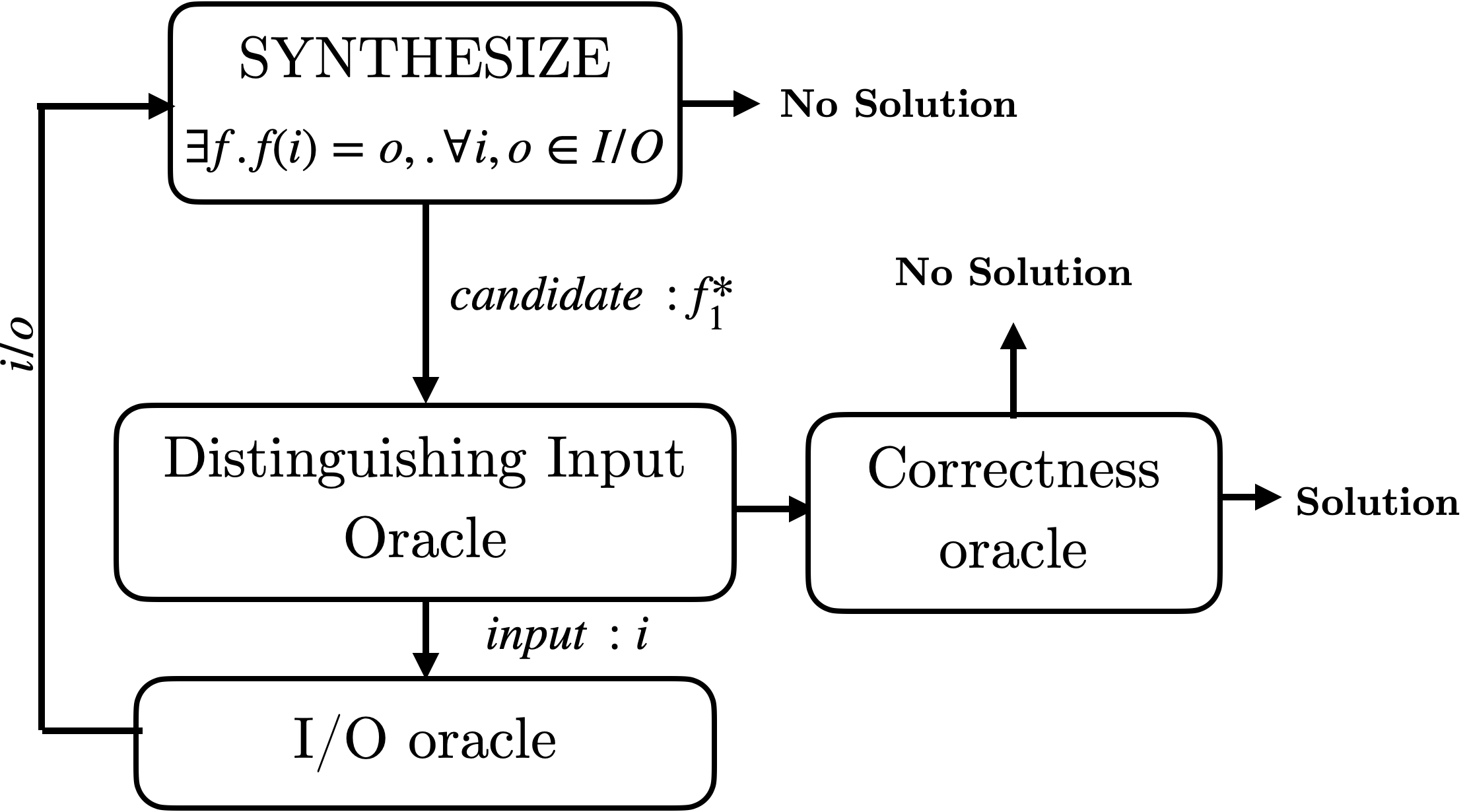}
    \caption{Synthesis with distinguishing inputs}
    \label{fig:ogiscomponent}
\end{figure}
The \symo algorithm described in this work will implement ICE learning when given a correctly defined set of oracles and oracle interface and a constraint $\theta_{corr}(inv)=\mathit{true}$. Interfaces for these oracles, in a system with variables $x_1 \ldots x_n$, are shown in Figure~\ref{fig:iceinterfaces}. 
Note that implication queries generate constraints enforcing that if the pre-state of the implication pair lies in the invariant, so must the post-state, which also allows the learner to exclude the pre-state in its next round of synthesis.
\begin{figure*}
\resizebox{.9\linewidth}{!}{
\begin{minipage}{\textwidth}
\begin{align*}
            I_{corr} &= 
            \begin{cases}
            Q&: (f^* \, (\sigma_1 \times \ldots \times \sigma_n \rightarrow bool))\\
            R&: (z_b \,\, bool)\\
            \alpha_{gen}&: \theta_{corr}(f^*)=z_b \\
            \end{cases}
            I_{neg} &= 
            \begin{cases}
            Q&: (f^* \, (\sigma_1 \times \ldots \times \sigma_n \rightarrow bool))\\
            R&: (z_1 \,\sigma_1), \ldots,(z_n \,\sigma_n)\\
            \beta_{gen}&: \neg f(z_1, \ldots z_n)\\
            \end{cases}
\end{align*}
\begin{align*}
            I_{impl} =& 
            \begin{cases}
            Q&: (f^* \, (\sigma_1 \times \ldots \times \sigma_n \rightarrow bool))\\
            R&: (z_1 \,\sigma_1), \ldots,( z_n \,\sigma_n),\\&\,\,\,\,(z'_1 \,\sigma_1), \ldots,( z'_n \,\sigma_n)\\
            \beta_{gen}&: f(z_1, \ldots z_n) \, \implies  \\& \,\,\,\, f(z'_1, \dots z'_n)\\
%            \beta{gen}&: inv(z_1, \ldots z_n) \Rightarrow \\& \,\,\,\,inv(z'_1, \dots z'_n)\\
            \end{cases}
            I_{pos} =& 
            \begin{cases}
            Q&: (f^* \, (\sigma_1 \times \ldots \times \sigma_n \rightarrow bool))\\
            R&: (z_1 \,\sigma_1), \ldots,( z_n \,\sigma_n)\\
            \beta_{gen}&: f(z_1, \ldots z_n)\\
            \end{cases}
\end{align*}
\end{minipage}
}
\caption{Oracle interfaces for ICE learning, synthesizing an invariant $f$, and receiving a candidate invariant $f^*$. \label{fig:iceinterfaces}}
\end{figure*}

% \todo[inline]{describe CEGIS with external oracle}
% \begin{figure}
%     \centering
%     \includegraphics[width=0.5\textwidth]{interface-cegis.png}
%     \caption{Oracle interface for CEGIS with an external oracle}
%     \label{fig:my_label}
% \end{figure}
\subsubsection{Synthesis with distinguishing inputs}
This algorithm~\cite{ogis}, illustrated in Figure~\ref{fig:ogiscomponent}, uses several oracles which interact with each other. The synthesis phase searches for a function that satisfies a list of input-output examples. If one is found, it is passed to a \emph{distinguishing-input oracle}, which looks for another, different, function that behaves the same as the existing function on the list of input-output examples, but behaves differently on another  distinguishing input.

%don't know why this new page was needed but bad things happen without it
\newpage
 If such a function exists, the distinguishing input is passed to an input/output oracle, and the input/output pair is passed to the synthesizer. If a distinguishing input does not exist, the correctness of the function is checked and the algorithm terminates (if this function is not correct, then there is no solution to the synthesis problem).

% \begin{figure}
%     \centering
%     \includegraphics[width=0.5\textwidth]{interface-ogcps.png}
%     \caption{Oracle interface for oracle guided component based synthesis}
%     \label{fig:my_label}
% \end{figure}
We can implement this algorithm in \symo by observing that the synthesizer needs only to query the correctness oracle and the distinguishing input oracle, which maintains a list of inputs it has returned so far and receives a response from only one oracle, the input/output oracle. The interface for the distinguishing input-oracle is as follows for a specification synthesizing $f$, and a candidate synthesis function $y:\sigma_1\times \ldots \times \sigma_n \rightarrow \sigma'$:
            \[
            I_\mathit{synthDI}  = 
            \begin{cases}
            Q&: (y \,(\sigma_1 \times \ldots \times \sigma_n \rightarrow \sigma'))\\
            R&: (z_1 \,\sigma_1), \ldots,( z_n \,\sigma_n), (z_{n+1} \, \sigma')\\
            \beta_{gen}&: f(z_1, \ldots z_n) = z_{n+1}\\
            \end{cases}
            \]

\subsubsection{CEGIS(T)}
This algorithm~\cite{cegist} extends a CEGIS loop to incorporate a theory solver that, given a candidate solution, generalizes the candidate to a template solution (removing any constant literals in the solution and replacing them with symbolic values), and searches for solutions within that template. If a solution is found, the oracle returns that valid solution. If no such solution exists, a constraint is passed back that blocks the candidate solution. 
Our framework can be used to implement an approximation of CEGIS(T): one in which the oracle generalizes the candidate and searches for a solution within that template, but if no solution is found, a single point counterexample is returned instead of the constraint over the syntactic space. This highlights a limitation of the framework we propose: the oracles may only place \emph{semantic} constraints over the search space. That is, it may place constraints that reason about the \emph{behavior} of the synthesis functions, and not the syntactic space from which the synthesis functions are to be constructed. The oracle interface for our modified CEGIS(T) is as follows:
            \[
            I_{cegist}  = 
            \begin{cases}
            Q&: (y \,\, (\sigma_1 \times \ldots \times \sigma_n \rightarrow \sigma'))\\
            R&: (z_1 \,\sigma_1), \ldots,( z_n \,\sigma_n),(z_{n+1} \,\, b),\\&\,\, (z_{n+2} \,\, (\sigma_1 \times \ldots \times \sigma_n \rightarrow \sigma'))\\
            \beta_{gen}&: z_{n+1}\,?\,(f = z_2): \syntformula(z_1/x_1 ,\ldots z_n/x_n) 
            \end{cases}
            \]

\subsection{Evaluation Details}
\label{sec:results_extended}
In this section, we give more details on the control, mathematics, and image processing case studies. We also provide a brief practical comparison of invariant synthesis with CEGIS and with ICE learning.

\subsubsection{Further details on case studies }
\paragraph{Control benchmarks}
The control benchmarks~\cite{control-cegis} synthesize single- and double-point precision floating-point controllers that guarantee stability and bounded safety for Linear Time Invariant systems. We use a state-space representation, which is discretized in time with $6$ different constant sampling intervals $T_s$, generating $6$ benchmarks per system:

\begin{align*}
\label{eq:ode1}
\dot{x}_{t+1} = \mat{A}\vec{x}_t+ \mat{B} \vec{u}_t, 
\end{align*}
where 
$\vec{x} \in \mathbb{R}^n$,  
$\vec{u} \in \mathbb{R}^p$ is the input to the system, calculated as $\mat{K}\vec{x}$ where $\mat{K}$ is the controller to be synthesized, 
$\mat{A} \in \mathbb{R}^{n \times n}$ is the system matrix, 
$\mat{B} \in \mathbb{R}^{n \times p}$ is the input matrix,
and subscript $t$ indicates the discrete time step. 

%soundly 
For stability benchmarks, we aim to find a stabilizing controller, such that absolute values of the (potentially complex) eigenvalues of the closed-loop matrix $\mat{A} - \mat{B}\mat{K}$ are less than one. For bounded safety benchmarks, we aim to find a controller that is both stable as before and guarantees the states remain within a safe region of the state space up to a given number of time steps. 

The \symo encoding uses an oracle to determine the stability of the closed-loop matrix. The encoding without oracles requires the SMT solver to find roots of the characteristic polynomial.

\paragraph{Educational mathematics bencharks} These benchmarks are taken from Edexcel mathematics questions~\cite{maths1}. The questions require the SMT solver to find numbers that are (some combination of) factors, prime-factors, square and triangle numbers. The encodings without oracles used recursive functions to determine whether a number is a prime or a triangle number. We note the oracle used alongside the benchmark number in Table~\ref{tab:maths}. We enable the techniques described by Reynolds et. al.~\cite{cvc4recfun} when running CVC4 on problems using recursive functions.
\begin{table}[]
    \centering
    \begin{tabular}{l |c|c | c }
     & \toolname & \toolname & CVC4 \\ 
    benchmark & (oracles) & (no oracles) & (no oracles) \\\hline 
    1b-square & $<$0.2s & $<$0.2 &  -- \\
    1d-prime & $<$0.2s & -- & $<$0.2s \\
    1f-prime & 3.1s & --& $<$0.2s \\
    1h-triangle & $<$0.2s & --& $<$0.2s \\
    1j-square,prime & $<$0.2s & -- & -- \\
    1l-triangle & $<$0.2s & --& $<$0.2s \\
    1m-triangle & $<$0.2s & --& $<$0.2s \\
    ex7-prime & 2.3s & --& -- \\
    ex8-prime & --& --& -- \\
    ex9-prime & 3.2s& --& -- \\
    ex10-prime & --& --& -- \\
    ex11-prime & $<$0.2s& --& -- \\ \hline
    \end{tabular}
    \caption{Solving times for \toolname and CVC4 on math examples using oracle and recursive function encodings. `` -- '' indicates the timeout of $600$s was exceeded.}
    \label{tab:maths}
\end{table}

\paragraph{Image transformations} We run 9 image transformations on images up to $256$x$256$ pixels. The oracle-free encodings consist of an equality constraint for each pixel, and so are more than $50,000$ constraints long. For the oracle-free encoding we simplify the problem by assuming that the images are given as a raw matrix, and this encoding, unlike the oracle-based encoding, does not include the JPEG file format decoder. The results are shown in Table~\ref{tab:image_res}. 
\begin{table}[h]
    \centering
    \begin{tabular}{l |c |c | c  }
         & \toolname &\toolname& CVC4 \\
     Benchmarks      & {(oracles)}&{(no oracles)}&{(no oracles)} \\
        \hline
        Image invert  & 0.4s & -- & --\\  
        Image zombie & 0.4s & -- & --\\ 
        Image attenuate & 0.9s & --& -- \\
        Image crop1  & 14s  & -- & -- \\ 
        Image crop2  & 11.8s & -- & -- \\ 
        Image crop3  & 12.0s  & -- & --  \\ 
        Image brighter & 35.0s & -- & -- \\
        Image darker  & 120.0s  & -- & -- \\
        Image round  & 0.2s & -- & -- \\
        Image blur  & -- & -- & -- \\
    \end{tabular}
    \caption{Solving time for image transformation examples, comparing the oracle vs oracle-free encoding. `` -- ''  indicates the timeout of $600$s was exceeded. The oracle-free encodings are $>$50,000 lines long.}
    \vspace{-1mm}
    \label{tab:image_res}
\end{table}

\subsubsection{A practical comparison of CEGIS vs ICE}
As an illustration of \smo performing ICE-learning, we implement positive, negative, and implication oracles for a small set of invariant benchmarks taken from the SyGuS competition~\cite{sygus-comp}. We note that, for positive witness oracles to be notably informative beyond a pure specification, the positive witness oracle should unroll the transition relation. This is easier with deterministic transition relations. 
Table~\ref{tab:mini_inv} shows a comparison of CEGIS against ICE-learning with positive witness oracles that unroll the transition relation (in Inv-det). We observe that using such oracles, ICE-learning can solve these examples faster and using fewer generated synthesis constraints. Note that we disable all heuristics for solving invariants for these examples.
 We also note that the failure mode of the respective algorithms differs: CEGIS typically generates increasingly many constraints, enumerating through constant values; whilst ICE generates sets of constraints that become tricky to solve and the synthesis phase absorbs all the solving time. 
 Unsurprisingly, in the case where invariant heuristics are enabled and the positive witness oracles are not able to unroll the transition relation (Inv-all), the specifically tailored oracles do not provide a performance gain as they aren't able to provide information that isn't already contained within the specification used by CEGIS.

\begin{table}
    \centering
    \begin{tabular}{l | c | c  |c }
        Benchmarks & config. & \# solved &   t \\
        \hline
        Inv-det(5) & CEGIS & 5  &  163s  \\
        Inv-det(5) & ICE & 5  & 121s\\
        Inv-all(127) & CEGIS & 102 & 0.8s  \\
        Inv-all(127) & ICE & 86  & 2.1s\\
    \end{tabular}
    \caption{Comparison of ICE learning with CEGIS in \toolname. Inv-det have deterministic transition relations and ICE learning uses a positive witness oracle that can unroll the transition relation. In the Inv-all category, all transition relations are treated as non-deterministic and the positive witness oracle is SMT-based only. $t$ is the average solving time per benchmark (excluding timeouts, where the timeout is $600$s).}
    \label{tab:mini_inv}
\end{table}
\end{document}